\documentclass[journal]{IEEEtran}

\usepackage{graphicx}
\usepackage{amsmath,amssymb,cite}
\usepackage{bbold}
\usepackage{algpseudocode}
\usepackage{algorithm}
\usepackage{caption}
\usepackage{subcaption}
\usepackage{amssymb,bm}
\usepackage{cite}
\usepackage{amsthm}
\usepackage{algpseudocode}
\usepackage{epstopdf}

\usepackage{trackchanges}
\addeditor{JL}

\newtheorem{assumption}{Assumption}
\newtheorem{theorem}{Theorem}

\newtheorem{lemma}{Lemma}

\begin{document}
%
\title{Consensus-based Distributed Quantile Estimation in Sensor Networks}
%
%
%

\author{Jongmin~Lee,~\IEEEmembership{}
        Cihan~Tepedelenlioglu,~\IEEEmembership{Senior~Member,~IEEE,}
        and~Andreas~Spanias,~\IEEEmembership{Fellow,~IEEE}
        \vspace{-5mm}
\thanks{This work was supported in part by the NSF award ECCS 1307982 and the SenSIP Center. The authors are with the School of Electrical, Computer, and Energy Engineering Arizona State University, Tempe, AZ 85287, USA e-mail: jongmin.lee@asu.edu; cihan@asu.edu; spanias@asu.edu.}
}

\maketitle

\begin{abstract}
A quantile is defined as a value below which random draws from a given distribution falls with a given probability. In a centralized setting where the cumulative distribution function (CDF) is unknown, the empirical CDF (ECDF) can be used to estimate such quantiles after aggregating the data. In a fully distributed sensor network, however, it is challenging to estimate quantiles. This is because each sensor node observes local measurement data with limited storage and data transmission power which make it {}{difficult to} obtain the global ECDF. This paper proposes consensus-based quantile estimation for such a distributed network. The states of the proposed algorithm are recursively updated with two-steps at each iteration: one is a \emph{local update} based on the measurement data and the current state, and the other is \emph{averaging} the updated states with neighboring nodes. We consider the realistic case of communication links between nodes being corrupted by independent random noise. It is shown that the estimated state sequence is asymptotically unbiased and converges toward the sample quantile in the mean-square sense. The two step-size sequences corresponding to the averaging and local update steps result in a mixed-time scale algorithm with proper decay rates in order to achieve convergence. We also provide applications to distributed estimation of trimmed mean, computation of median, maximum, or minimum values and identification of outliers through simulation.  
\end{abstract}

\begin{IEEEkeywords}
consensus, distributed quantile estimation, wireless sensor networks, measures of central tendency.
\end{IEEEkeywords}

%
\IEEEpeerreviewmaketitle

\section{Introduction} \label{introduction}

\IEEEPARstart{D}{istributed} sensors measure physical phenomena observable over a certain region and fuse the sensed information by local communications. This type of network is scalable and energy efficient because each node shares its data only with neighbors. A traditional problem in this domain is to estimate the sample average of measurements by iteratively averaging the states with neighboring ones, and achieve a \emph{consensus} on the global average of the initial measurements \cite{Boyd2004,Xiao2005,Olfati2007}. This has influenced many distributed estimation applications due to the broad use of the arithmetic mean in signal processing techniques.

Distributed average consensus of sensor measurement data can be used in monitoring applications. One example is monitoring average temperature (or, other statistical metrics) over a sensor network in remote areas. The arithmetic mean of temperature data represents the \emph{central tendency} of temperature. However, the mean can be vulnerable, as a measure of central tendency, to the skewness of the distribution. Outliers can also cause bias to the sample mean. An alternative metric is the median that represents the midpoint which divides the ordered dataset into two subsets of equal size. More generally \emph{quantiles} are the generalized inverse of the CDF at a certain probability. Beyond estimating the median, quantiles can be used in various applications such as outlier removal and computation of robust statistics from a set of measurement data by eliminating the values higher (or lower) than a certain cutpoint. One such robust statistic is the trimmed mean which is an average of the data excluding outliers. Maximum and minimum values can be viewed as extreme examples of quantiles. Quantile regression estimates the conditional quantiles of measurement data distribution where the statistics such as mean and variance may change over time. This method has been used in a variety of machine learning \cite{Meinshausen2006} as well as statistical applications \cite{Koenker2005}.

\begin{figure}
	\centering
	\vspace{-0mm}
	\includegraphics[width=0.9\linewidth]{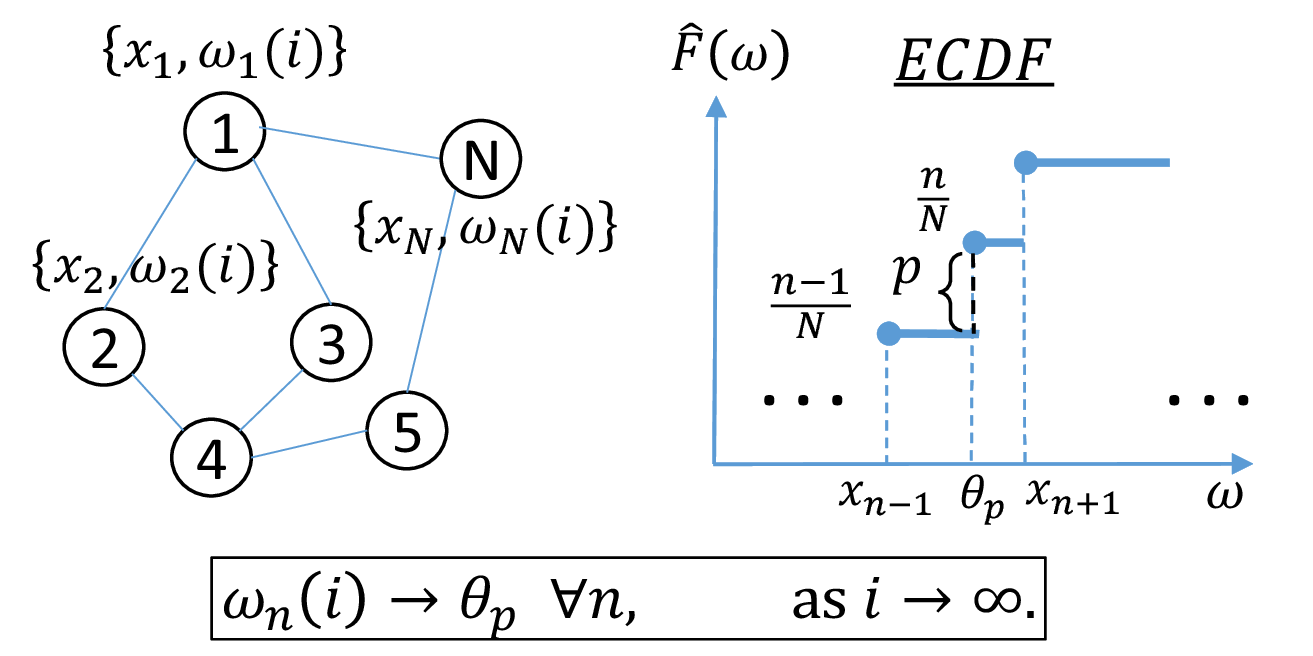}
	\caption{Overview of distributed quantile estimation. State at every node converges to a desired quantile $\theta_p$ defined as \eqref{eq:theta_p}.}
	\vspace{-6mm}
	\label{fig:system_model}
\end{figure}

In this paper, we consider the quantile estimation problem in a distributed setting, as illustrated in Fig. \ref{fig:system_model}. This is necessary if nodes in a distributed network have local measurement data only but want to know the quantile value without the knowledge of the global aggregate CDF estimate at each node. The sensor network is assumed to be fully distributed where there is no fusion center and the sensors are unlabeled. Each node maintains its own data and state of estimate, and communicates the information only with neighboring nodes via noisy communication links between nodes. Detailed knowledge of the network graph structure beyond neighbors is inaccessible to every node. The states of quantile estimates are recursively updated with two steps at each iteration. The \emph{local update} step is based on the individual measurement data and the current state of the quantile estimate. The updates are transfered to the neighboring nodes by \emph{averaging} the estimates. We analyze convergence behavior by showing that the estimated state sequence is asymptotically unbiased and converges toward the true quantile in the mean-square sense. The proposed algorithm is applicable for identifying outliers and calculating median as well as trimmed mean.


There are significant number of works related to consensus-based distributed parameter estimation. See \cite{Tsitsiklis1984,Tsitsiclis1986,Bertsekas1989,Yin1991} for the early works, which inspired numerous applications. Distributed least-mean-square (LMS) algorithm is introduced to estimate a linear system parameter in various scenarios \cite{Sayed2008,Sayed2013bookch,Giannakis2009,Stankovic2011,Lee2015}. In these works, sensors observe random data at very iteration, generated by a linear system with a parameter vector. In \cite{Kar2012}, the authors proposed the consensus plus innovation scheme for distributed parameter estimation with single- and mixed-time scales. They consider nonlinear as well as linear system models and show convergence analysis. They assume that the sensors observe random data at every iteration and the observation model is continuous and invertible. In contrast, our model uses the empirical CDF (ECDF) which is discontinuous and non-invertible. Our work may be considered as a root finding problem which is similar with Robbins-Monro stochastic approximation algorithm \cite{Robbins1951}, but we consider a distributed graph network setup. Reference \cite{Bianchi2013} shows a performance analysis for Robbins-Monro algorithm in a distributed framework, where they considered the asynchronous random gossip algorithms \cite{Boyd2006RGA} with random data observation of a continuous function at every iteration. However, our work assumes that the size of measurement data is finite, utilizing ECDF which is nonlinear, discontinuous, and non-invertible.

For a limited size of measurement data in distributed networks, our work can be considered as solving a distributed node selection problem, which is to find the $n$-th smallest measurement out of $N$ data samples which can be related to $n$-th quantile. The references \cite{Kempe2003,Kuhn2007,Negro1997,SANTORO1992} are a few most closely related to ours. They are similar with our work in that each node maintains a piece of the entire data set and quantile state information, wishing to identify $n$-th smallest data. However, their main contributions are fundamentally different than ours. In \cite{Kempe2003}, a leader node is chosen to maintain candidates of quantiles at each round of the gossip protocol, reducing the number of candidates until only a single candidate is left. In \cite{Kuhn2007}, it is assumed that every node knows the network's diameter which is defined as the length of the ``longest shortest path'' between any two sensor nodes. However, their algorithm needs to maintain a set of candidates for the $n$-th selection steadily reducing the set until it reaches the desired element under a certain criterion. Reference \cite{Negro1997} depends on guessing and selection strategy to find the $n$-th smallest element. Their algorithm maintains a set of control messages such as \emph{start}, \emph{small}, \emph{large}, and \emph{stop} where the messages are transmitted to the entire network at every communication iteration. A distributed selection algorithm in \cite{SANTORO1992} is based on broadcasting the control messages, and increases the number of message exchanges as the network size becomes larger. In contrast, our algorithm is a fully distributed method without any type of leading nodes or candidate sets. The algorithm is also scalable because control messages are not transmitted to every node. Furthermore, our work considers the more realistic case of communication links between nodes being corrupted by independent random noise, whereas the references above are based on noiseless communication links.

This paper consists of the following. In Section \ref{graph_theory} we review graph theory. Section \ref{system_model} and \ref{problem_statement} describe system model and problem statements followed by the proposed algorithm in Section \ref{distributed_quantile_estimation}. Convergence analysis is provided in Section \ref{convergence_analysis}. We illustrate the proposed algorithm with simulations in Section \ref{simulations}. We describe conclusions in Section \ref{conclusion}.

\vspace{-1mm}
\section{Network Graph Theory}\label{graph_theory} 

There is an undirected graph $\mathbb{G} = \left(\mathbb{N},\mathbb{E}\right)$ containing a set of nodes $\mathbb{N} = \left\{1,\ldots,N\right\}$ and a set of edges $\mathbb{E}$. The neighbors of node $n$ is denoted by $\mathbb{N}_n = \left\{ l | \left\{n,l\right\} \in \mathbb{E}\right\}$ where $\left\{n,l\right\}$ is an edge between the nodes $n$ and $l$ \cite{Chung97a}. Each node communicates with neighbors via the edges. The degree $d_n$ denotes the number of neighbors at node $n$, and $d_{\text{max}}$ the maximum degree. A graph is \emph{connected} if there exists at least one path between every pair of nodes. The graph structure is described by adjacency matrix $\mathbf{A} = \left\{a_{nl}\right\}$, which is an $N \times N$ symmetric matrix, whose element in the $n$-th row and $l$-th column is $a_{nl} = 1$ if $\left\{n,l\right\} \in \mathbb{E}$. The diagonal matrix $\mathbf{D} = \mathrm{diag}\left[d_1,d_2,\ldots,d_N\right]$ represents the degrees of all the nodes. The Laplacian matrix is given by $\mathbf{L} = \mathbf{D} - \mathbf{A}$, where the eigenvalues of $\mathbf{L}$, denoting $\lambda_{n}(\mathbf{L})$ for $n \ge 1$, are non-negative and the number of zero eigenvalues can be shown to be the number of distinct components of the graph. There are, when the graph is connected, $\lambda_{1} = 0$ and $\lambda_{n} > 0$ for $n \ge 2$ so that the rank of $\mathbf{L}$ for a connected graph is $N-1$. The vector $\mathbf{1}$ {}{of all $1$'s} is the eigenvector of $\mathbf{L}$ associated with the eigenvalue $0$, i.e., $\mathbf{L}\cdot\mathbf{1} = 0\mathbf{1}$. The eigenvalue $\lambda_{2}$ is called the algebraic connectivity and characterizes how densely the graph is connected.

\vspace{-1mm}
\section{System Model} \label{system_model} 
\normalsize
Consider $N$ sensor nodes over a connected and undirected distributed network $\mathbb{G} = \left(\mathbb{N},\mathbb{E}\right)$ where there is no fusion center. Due to the connectedness, the eigenvalue $\lambda_{2}$ of the Laplacian matrix $\mathbf{L}$ is positive. Each node $n$ has a scalar measurement denoted by $x_n \in \mathbb{R}$, where $n=1,\ldots,N$, and $\{x_n\}_{n=1}^N$ defines ECDF as follows:

\vspace{-4mm}
\small
\begin{align}
\widehat{F}(\omega;\mathbf{x}) = \frac{1}{N}\sum_{n=1}^{N} u(\omega - x_n)   \label{eq:p_ecdf}
\end{align}
\normalsize
where the step function $u(\cdot)$ is given by

\vspace{-4mm}
\small
\begin{align}
u(\omega - x_n) = \left\{\begin{array}{cc}
1, & \text{if }  \omega \ge x_n  \\ 
0, & \text{otherwise}
\end{array}\right. 
\end{align}
\normalsize
{where $N$ is known to every node $n$, which can be obtained by node counting algorithms such as \cite{SZhang2017}.}

Without loss of generality, it can be assumed that the measurement data set is sorted in ascending order. Let $\mathbf{x} = \left[x_1,\ldots,x_N\right]^T$ where $x_1 \le \cdots \le x_N$. Each node maintains a real-valued scalar state to be updated for quantile estimation. Let $\omega_n(i)$ denote the state of node $n$ at time $i$. The state is transferred to neighboring nodes via wireless links in the presence of random communication noise $\xi_{nl}(i)$ from node $l$ to $n$. Random noise on the link from $l$ to $n$ is assumed independent and identically distributed (i.i.d.) random process $\{\xi_{nl}(i)\}_{i \ge 0}$ with zero mean and variance $E\left[\xi_{nl}^2(i)\right]$ where $\sup_{n,l,i} E\left[\xi_{nl}^2(i)\right] < \infty$. As the communication iteratively continues, node $n$ updates its own state $\omega_n(i)$ based on its own measurement $x_n$ and neighbors' states $\{\omega_l(i)\}_{l \in \mathbb{N}_n}$ where $\mathbb{N}_n$ denotes the set of neighboring nodes of $n$. 

Let $0 < p < 1$ denote the probability that corresponds to a quantile $\theta_p$. When $p = 0.5$, the corresponding quantile $\theta_{0.5}$ is the median of $\mathbf{x}$. When $p=0.75$, the corresponding $\theta_{0.75}$ indicates that 75\% of measurement data is less than or equal to $\theta_{0.75}$. Note that the ECDF $\widehat{F}(\omega;\mathbf{x})$ in \eqref{eq:p_ecdf} is a stair-case function, and $\theta_p$ is {}{a generalized} inverse of the ECDF in some {}{appropriate} sense. More formally, for the ECDF $\widehat{F}(\omega,\mathbf{x})$, the relation between $p$ and $\theta_p$ can be defined \cite{Hyndman1996} as

\vspace{-4mm}
\small
\begin{align}
\theta_p =  \inf_\omega \left\{\omega: \widehat{F}(\omega;\mathbf{x}) \ge p \right\}. \label{eq:theta_p}
\end{align}
\normalsize
%
%

Quantiles may be centrally obtained by using the ECDF of \eqref{eq:p_ecdf} after collecting all the measurement data $\bf{x}$. Practically in distributed wireless sensor networks, however, the centralized method is not directly applicable.

\vspace{-1mm}
\section{Problem Statement} \label{problem_statement} 
Since each node has only a single measurement data, it may be impossible to know the global ECDF $\widehat{F}(\omega;\mathbf{x})$ in large-scale networks. In addition, it is difficult to synchronize the local states of all nodes (i.e., having $\{\omega_n(i)\}_{n=1}^N$ to be $\omega(i)$ for all $n$) at every iterative update over the network. The centralized method may require transmission of the measurement data $\mathbf{x}$ and the states $\{\omega_n(i)\}_{n=1}^N$ from all nodes to a fusion center with undesirable transmission power consumption. Also, all the information exchange is corrupted by communication random noise. Despite the constraints mentioned above, we want every node $n$ to estimate the quantile $\theta_p$ for a given {}{$0<p<1$} as $i \to \infty$.

Suppose the ECDF of \eqref{eq:p_ecdf} in a fully distributed network $\mathbb{G} = \left(\mathbb{N},\mathbb{E}\right)$. There is no fusion center to collect the measurement data. Each node $n$ communicates within neighborhood $\mathbb{N}_n$ via wireless communication channel corrupted by random noise, as described in Section \ref{system_model}. Given $x_n$ and $p \notin \{\frac{1}{N}, \frac{2}{N} \ldots, 1\}$ at node $n$ (i.e., $p$ does not correspond to a point of discontinuous ECDF), we want a distributed quantile estimation algorithm that generates the state $\omega_n(i)$ such that, as $i\to \infty$,

\vspace{-4mm}
\small 
\begin{align} 
\omega_n(i) \to \theta_p,\,\,\,\,\,\forall n. \label{eq:omega_ni_probelmstatement}
\end{align}
\normalsize 
We consider without loss of generality $p \notin \{\frac{1}{N}, \frac{2}{N} \ldots, 1\}$ in this paper. If instead $p=\frac{n}{N}$, estimated quantile can be shown to be any value within the interval $\left[x_n,x_{n+1}\right)$.

\vspace{-2mm}
\section{Distributed Quantile Estimation} \label{distributed_quantile_estimation}
A consensus-based distributed algorithm is proposed where for a given $p$ each node $n$ locally updates $\omega_n(i)$ that satisfies \eqref{eq:omega_ni_probelmstatement}{}{.} Let $\omega_n(i)$ and $\psi_n(i)$ denote respectively the state of quantile estimate and an intermediate state variable at iteration $i$. Node $n$ updates its state $\omega_n(i)$ based on the local measurement data $x_n$ for the given constant $p$. {}{The algorithm consists of a \emph{local update} of the intermediate variable $\psi_n(i)$ followed by an \emph{averaging} step where $\omega_n(i)$ is updated.} The \emph{local update} step is given by

\vspace{-4mm}
\small
\begin{align}
\psi_n(i) = \omega_n(i) - \alpha(i) \Big[u\big(\omega_n(i)-x_n\big) - p\Big],\,\,\,\,\forall n, \,\,i\ge0, \label{eq:local_update_approx_u}
\end{align}
\normalsize
where $\{\alpha(i)\}_{i \ge 0}$ is a deterministic step-size sequence that will be explained later in detail. The local update step of \eqref{eq:local_update_approx_u} at node $n$ {}{involves} with its own data $x_n$. The \emph{averaging} step at node $n$ is then performed by

\vspace{-4mm}
\small
\begin{align}
\omega_n(i+1) = \psi_n(i) - \eta(i) \sum_{l \in \mathbb{N}_n} \Big[\psi_n(i) - \big(&\psi_l(i) + \xi_{nl}(i)\big) \Big], \label{eq:avg_step}
\end{align}
\normalsize
$\forall n$, $i \ge 0$, where $\psi_l(i)$ denotes the state transmitted from node $l$ with being perturbed at node $n$ by communication random noise $\xi_{nl}(i)$,  $\mathbb{N}_n$ denotes the set of neighbors of node $n$, and $\eta(i)$ is the step-size that controls exchange rate of node $n$ with neighboring nodes at time $i$. We consider a deterministic sequence $\{\eta(i)\}_{i\ge0}$ that will be explained later in this section. 

We now describe the algorithm in \eqref{eq:local_update_approx_u} and \eqref{eq:avg_step} in vector form. Let $\bm{\omega}(i) = \left[\omega_1(i),\ldots,\omega_N(i)\right]^T$ and $\bm{\psi}(i) = \left[\psi_1(i),\ldots,\psi_N(i)\right]^T$. Laplacian matrix $\mathbf{L}$ is described in Section II. Equations \eqref{eq:local_update_approx_u} and \eqref{eq:avg_step} can be {}{expressed} as

\vspace{-4mm}
\small
\begin{align}
\bm{\psi}(i) &= \bm{\omega}(i) - \alpha(i) \mathbf{y}(i), \label{eq:local_update_approx_u_vector} \\
\bm{\omega}(i+1) &= \big(\mathbf{I} - \eta(i) \mathbf{L}\big) \bm{\psi}(i) - \eta(i) \bm{\xi}(i), \label{eq:avg_step_vector}
\end{align}
\normalsize
where

\vspace{-4mm}
\small
\begin{align}
\mathbf{y}(i)&=\big[y_1(i),\ldots,y_N(i)\big]^T, \\
y_n(i) &\triangleq u\big(\omega_n(i) - x_n\big) - p,\,\,\,\,\,\forall n, \label{eq:y_n(i)} \\
\bm{\xi} (i) &= - \Bigg[\sum\limits_{l \in \mathbb{N}_1} \xi_{1l}(i), \ldots, \sum\limits_{l \in \mathbb{N}_N} \xi_{Nl}(i)\Bigg]^T. \label{eq:xi(i)}
\end{align}
\normalsize
Combining \eqref{eq:local_update_approx_u_vector} and \eqref{eq:avg_step_vector}, we can express the distributed quantile estimation algorithm as, for $i \ge 0$,

\vspace{-4mm}
\small
\begin{align}
\bm{\omega}(0) &= \mathbf{x},\nonumber \\
\bm{\omega}(i+1) &= \big(\mathbf{I} - \eta(i)\mathbf{L}\big)\big(\bm{\omega}(i)- \alpha(i)\mathbf{y}(i)\big) - \eta(i) \bm{\xi}(i).\label{eq:combined_averaging_and_local_update_w_noise}
\end{align}
\normalsize

The step-sizes satisfy the persistence condition:

\vspace{-4mm}
\small 
\begin{align}
&\alpha(i) > 0, \,\,\sum_{i=0}^{\infty}\alpha(i) = \infty,\,\, \sum_{i=0}^{\infty} \alpha^2(i) < \infty, \label{eq:persistent_alpha}\\
&\eta(i) > 0, \,\,\sum_{i=0}^{\infty}\eta(i) = \infty,\,\, \sum_{i=0}^{\infty} \eta^2(i) < \infty. \label{eq:persistent_eta}
\end{align}
\normalsize
The conditions \eqref{eq:persistent_alpha} and \eqref{eq:persistent_eta} imply that decaying rates of the step-sizes are fast but not too fast. This condition has been commonly used for convergence analysis, based on conventional stochastic approximation theory \cite{Nevelson1973,Spall2003,Chen2006}. However, the distributed quantile estimation algorithm \eqref{eq:combined_averaging_and_local_update_w_noise} is a combined vector form of \eqref{eq:local_update_approx_u_vector} and \eqref{eq:avg_step_vector}, and results in a mixed-time scale for the iterative updates in \eqref{eq:combined_averaging_and_local_update_w_noise}. {}{For $\bm{\omega}(i)$ to converge as $i \to \infty$, the step-size $\alpha(i)$ needs to decrease faster than $\eta(i)$. Rewriting the algorithm of \eqref{eq:combined_averaging_and_local_update_w_noise}, we have for $i \ge 1$}

\vspace{-4mm}
\small
\begin{align}
{}{
\bm{\omega}(i+1) = \bm{\omega}(i) - \eta(i)\mathbf{L}\big(\bm{\omega}(i)  - \alpha(i)\mathbf{y}(i)\big) - \alpha(i) \mathbf{y}(i) - \eta(i)\bm{\xi}(i).} \label{eq:combined_averaging_and_local_update_w_noise_stochastic_theory}
\end{align}
\normalsize
{}{The output element of vector $\mathbf{y}(i)$ is $-p$ or $1-p$ by the definition \eqref{eq:y_n(i)}. $\mathbf{L}\bm{\omega}(i)$ in \eqref{eq:combined_averaging_and_local_update_w_noise_stochastic_theory} indicates an average state vector of neighboring nodes. If the decaying rate of $\alpha(i)$ was slower than that of $\eta(i)$, $\bm{\omega}(i)$ would not converge because $\alpha(i)\mathbf{y}(i)$ contributes to $\bm{\omega}(i+1)$ more significantly than $\eta(i)\mathbf{L}\bm{\omega}(i)$, as iteration continues. A convergence analysis for such a mixed-time scale approach with appropriate choices of step-sizes was also used in \cite{Kar2012}. Nevertheless, $\mathbf{y}(i)$ plays an important role in estimating the true quantile. $\mathbf{L}\mathbf{y}(i)$ in \eqref{eq:combined_averaging_and_local_update_w_noise_stochastic_theory} is an approximate ECDF minus $p$ that can contribute to how much $\bm{\omega}(i+1)$ is changed from $\bm{\omega}(i)$ by $\mathbf{L}\bm{\omega}(i) - \alpha(i)\mathbf{L}\mathbf{y}(i)$. If $\mathbf{L}\mathbf{y}(i)$ is large, which means the approximate ECDF with the current state $\bm{\omega}(i)$ results in large error compared to the desired $p$, then $\bm{\omega}(i+1)$ needs to update significantly. As $\bm{\omega}(i)$ approaches the true quantile, $\mathbf{L}\mathbf{y}(i)$ becomes small, and $\bm{\omega}(i+1)$ changes small as well. }


We summarize the assumption for step-sizes that will be needed for the convergence.

\begin{assumption}
	(Decreasing step-sizes) The step-size $\alpha(i)$ in \eqref{eq:local_update_approx_u_vector} decreases faster than $\eta(i)$ in \eqref{eq:avg_step_vector} with the forms:
	
	\vspace{-4mm}
	\small
	\begin{align}
	\alpha(i) = \frac{\alpha_0}{(i+1)^{\tau_1}}\,\,\,\,\text{and}\,\,\,\,\eta(i) = \frac{\eta_0}{(i+1)^{\tau_2}},\,\,\,\,\text{for}\,\,i\ge 0, \label{eq:step_size_forms}
	\end{align}
	\normalsize
	where $\tau_1$ and $\tau_2$ denote constant decaying rates of $\alpha(i)$ and $\eta(i)$, respectively, $\alpha_0$ and $\eta_0$ are positive initial step-sizes, and $1 \ge \tau_1 > \tau_2 > 0.5$. Moreover, $\tau_1 - \tau_2 < 0.5$. \label{assumption:stepSizes}
\end{assumption}
One example of the step-size choice that satisfies Assumption \ref{assumption:stepSizes} is $\tau_1 = 1$ and $\tau_2 = 0.505$. As $\tau_1 - \tau_2$ decreases, difference between the decaying rates of $\alpha(i)$ and $\eta(i)$ also decreases, resulting in slowly decreasing {}{$\bm{\omega}(i)$ as stated in \eqref{eq:combined_averaging_and_local_update_w_noise_stochastic_theory}.}

\vspace{-2mm}
\section{Convergence Analysis} \label{convergence_analysis}
In this section we analyze convergence behavior of the distributed quantile estimation algorithm in \eqref{eq:local_update_approx_u} - \eqref{eq:avg_step}, or equivalently in \eqref{eq:combined_averaging_and_local_update_w_noise}. It is shown that the state sequence $\{\omega_n(i)\}_{i \ge 0}$ at node $n$ is asymptotically unbiased in Theorem \ref{theorem:unbiasedness} and the estimated sequence converges to the true quantile $\theta_p$ in mean-square sense in Theorem \ref{theorem:mean-square-convergence}. To achieve the above results, we use some properties of real number sequences described in Lemma \ref{lemma1} and bounded sequences in Lemma \ref{lemma:boundedness}.

{}
{Lemma \ref{lemma1} is used to show convergence property of the state sequence $\left\{\bm{\omega}(i)\right\}_{i \ge 0}$ with $\alpha(i)$ and $\eta(i)$ step sizes. Essentially, Lemma \ref{lemma1} is based on the stability study of a recursion form: $q(i+1) = r_1(i)q(i) + r_2(i)$ where $\left\{q(i)\right\}_{i \ge 0}$ is a state sequence and $r_1(i)$ and $r_2(i)$ are decreasing step size sequences.}
\begin{lemma}
	Consider the sequences $\{r_1(i)\}_{i\ge 0}$ and $\{r_2(i)\}_{i\ge 0}$, with non-negative constants $a_1$ and $a_2$,  which are given by
	
	\vspace{-4mm}
	\small
	\begin{align}
	r_1(i) = \frac{a_1}{\left(i+1\right)^{\delta_1}},\,\,\,\,\,r_2(i) = \frac{a_2}{\left(i+1\right)^{\delta_2}} \label{eq:lemma1_r1r2}
	\end{align}
	\normalsize
	where $0 \le \delta_1 \le 1$ and $\delta_2 \ge 0$. If $\delta_1 < \delta_2$, then, for arbitrary fixed $i_0$,
	
	\vspace{-4mm}
	\small
	\begin{align}
	\lim\limits_{i \to \infty} \sum_{k = i_0}^{i-1} \Bigg[\prod_{l=k+1}^{i-1}\big(1-r_1(l)\big)\Bigg] r_2(k) = 0.\label{eq:lemma1:0}
	\end{align}
	\normalsize
	\label{lemma1}
\end{lemma}
\begin{proof}\vspace{-3mm}
	See Lemma 25 in \cite{Kar2012}. 
\end{proof}

{}{Lemma \ref{lemma:boundedness} is used to show boundedness of sequences, combined with Lemma 1, in Theorem \ref{theorem:unbiasedness} and \ref{theorem:mean-square-convergence}.}
\begin{lemma}
	Define $\omega_{\text{avg}}(i) \triangleq \frac{1}{N} \mathbf{1}^T \bm{\omega}(i)$ that is the average of $\bm{\omega}(i)$ across all nodes at time $i$. Given the measurement data $\mathbf{x}$ and ratio $p$, the sequence $\{\eta(i)\}_{i \ge 0}$ of \eqref{eq:step_size_forms} satisfies
	
	\vspace{-4mm}
	\small
	\begin{align}
	\limsup_{i \to \infty} \eta(i) E \Big[\omega_{\text{avg}}(i) - \theta_p \Big] = 0, \label{eq:lemma:boundedness0} 
	\end{align}	\label{lemma:boundedness}
	\normalsize
\end{lemma}
\begin{proof}\vspace{-5mm}
	See Appendix \ref{sec:append_proof_lemma:boundedness}.
\end{proof}

{}{The proposed algorithm estimates the true quantile parameter. We show that the quantile estimation is asymptotically unbiased in Theorem \ref{theorem:unbiasedness} and it converges to the quantile in mean-square sense in Theorem \ref{theorem:mean-square-convergence}. The mean-square convergence indicates stronger consistency than convergence in probability.}

\begin{theorem}
	(Asymptotic Unbiasedness) Consider that a constant ratio $p$ is given for estimating a certain quantile $\theta_p$. Under Assumption \ref{assumption:stepSizes}, the state sequence $\{\omega_n(i)\}_{i \ge 0}$ in \eqref{eq:local_update_approx_u} and \eqref{eq:avg_step} at node $n$ is asymptotically unbiased:
	
	\vspace{-4mm}
	\small
	\begin{align}
	\lim_{i \to \infty} E\big[\omega_n(i)\big] = \theta_p\,\,\,\,\,\,\,\text{for }1\le n \le N.
	\end{align} \label{theorem:unbiasedness}
	\normalsize
\end{theorem}
\begin{proof}\vspace{-5mm}
	See Appendix \ref{sec:append_proof_theorem:unbiasedness}.
\end{proof}


\begin{theorem}
	(Mean-Square Convergence) Under Assumption \ref{assumption:stepSizes}, if $\sigma_\xi^2 < \infty$, then the sequence generated by the distributed quantile estimation algorithm \eqref{eq:combined_averaging_and_local_update_w_noise}, for a given ratio value $p$, converges to true quantile $\theta_p$ in mean-square sense:
	
	\vspace{-4mm}
	\small
	\begin{align}
	\lim_{i\to\infty} E\Big[\big\| \bm{\omega}(i) - \theta_p \mathbf{1}\big\|^2\Big] = 0, \label{eq:theorem2}
	\end{align}
	\normalsize  
	\label{theorem:mean-square-convergence}
\end{theorem}
\begin{proof} \vspace{-5mm}
	See Appendix \ref{sec:append_proof_theorem:mean_square_conv}. 	
\end{proof}

\vspace{-1mm}
{}{
From \eqref{eq:norm_sq_cauchy_schwarz_N_sigma2} in Appendix \ref{sec:append_proof_theorem:mean_square_conv}, we can observe that several factors affect convergence speed: communication noise variance $\sigma_\xi^2$, network size $N$, step size sequences $\alpha(i)$ and $\eta(i)$, and the distance between average state (denoted as $\mathbf{z}(i)$) and the quantile $\theta_p$. For example, if the decaying rates of $\alpha(i)$ and $\eta(i)$ are fast, the upper bound of $E\big[\| \bm{\omega}(i) - \theta_p \mathbf{1}\|^2\big]$ at $i$ is lowered. If the network size $N$ or the wireless link noise variance is large, the upper bound is also lowered. If the average state sequence, including the initial measurement $\mathbf{x}$ at $i=0$, is far from the true quantile $\theta_p$, then the convergence can be slow.}

\begin{figure}
	\centering
	\vspace{-5mm}
	\includegraphics[width=0.85\linewidth]{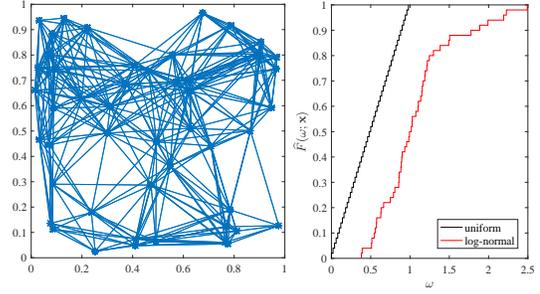}
	\vspace{-10mm}
	\caption{\emph{Left}: A graph for distributed sensor network. \emph{Right}: Empirical CDFs generated from $\{x_n\}_{n=1}^{N}$.}
	\vspace{-4mm}
	\label{fig:networkgeneration}
\end{figure}

\begin{figure}
	\centering
	\includegraphics[width=1\linewidth]{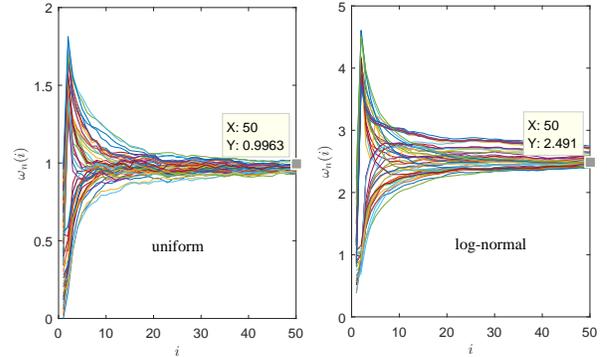}
	\caption{Maximum value estimations by setting $p = 0.99$ for $N=50$ in the presence of communication noise ($\sigma_{\xi}^2 = 0.09$). \emph{Left}: States $\left\{\bm{\omega}_n(i)\right\}_{n=1}^N$ converge toward $\theta_{p} = 0.98$ for the uniform data. \emph{Right}: $\theta_{p} = 2.4927$ for the log-normal data. Minimum can be also achieved by setting $p=0.01$.}
	\vspace{-9mm}
	\label{fig:maxconsensus_unif_lognormal}
\end{figure}

\vspace{-2mm}
\section{Simulations} \label{simulations}
In this section we demonstrate the distributed quantile estimation under various conditions. Consider a distributed sensor network, illustrated in Fig. \ref{fig:networkgeneration}, which is a connected graph with $N=50$ where the graph's connectivity is characterized by $\lambda_{2,\mathbf{L}} = 2.2815$. Note that this graph network is an example we do our experiments but other connected graph networks with different $N$ and $\lambda_{2,\mathbf{L}}$ can be used. Each node $n$ has a scalar measurement $x_n$ taken from a realization of random variable $X$ of either discrete uniform or log-normal distributions. The discrete uniform distribution is normalized, and the actual values are ranging from $0$ to $(N-1)/N$ with $1/N$ increase for each sample. The log-normal distribution was generated by $\ln X \sim \mathcal{N}(0,0.25)$. From all nodes, $N$ realizations of random variable $X$ whose ECDF is illustrated in Fig. \ref{fig:networkgeneration}. Without loss of generality, we can assume that the measurement data is distributed in ascending order: $x_1 \le x_2 \le \cdots \le x_N$. With the set of measurement data $\{x_n\}_{n=1}^N$, also denoted by $\mathbf{x}$ in vector form, a quantile $\theta_p$ of $\mathbf{x}$ is estimated for a desired ratio $p$ in a distributed way. The states $\{\omega_n(i)\}_{n=1}^N$ are recursively updated by the algorithm \eqref{eq:combined_averaging_and_local_update_w_noise}. The initial states $\bm{\omega}(0)$ are the nodes' own measurement data $\mathbf{x}$. Consider $ p = \frac{k-\varepsilon}{N}$ for $\theta_p = x_k$ where $\varepsilon = 0.5$ and $N$ is known to each node. Due to the choice of $p = \frac{k-\epsilon}{N}$, the $k$-th smallest element in $\mathbf{x}$ is estimated by achieving $\omega_n(i) = \theta_p$ for all $n$, as $i \to \infty$. We evaluate mean-squared error for convergence of the estimation by the following metric: 

\vspace{-4mm}
\small
\begin{align}
\frac{1}{N} E \Big[\big\|\bm{\omega}(i) - \theta_p \mathbf{1}\big\|^2\Big],\,\,i \ge 0, \label{eq:theorem2_metric}
\end{align}
\normalsize
where $\frac{1}{N}$ is due to normalization and $E\left[\cdot\right]$ can be approximated by ensemble averaging over 200 realizations. According to Assumption \ref{assumption:stepSizes}, we can set $\tau_1 = 1$ and $\tau_2 = 0.505$. We begin with $\alpha_0 = 1$ and $\eta_0 = 0.5/d_\text{max}$ for $\alpha(i)$ and $\eta(i)$ respectively, where $d_\text{max}$ denotes the maximum degree in graph network.


%


\vspace{-2mm}
\subsection{Distributed Quantile Estimation}
Given the sensor network $N=50$ and measurement data $\mathbf{x}$, suppose that $p=0.99$ is selected with $k=50$ and $\varepsilon=0.5$. Then the distributed algorithm \eqref{eq:combined_averaging_and_local_update_w_noise} estimates $\max(x_1,\ldots,x_N)$, as $i \to \infty$. Fig. \ref{fig:maxconsensus_unif_lognormal} shows that all the states converge toward $\theta_{p=0.99} = 0.98$, which is the maximum value of uniform $\mathbf{x}$ in the presence of communication noise. Similarly, one can estimate the minimum by setting $p=\frac{1-0.5}{50} = 0.01$. More generally, the $k$-th smallest element can be estimated by setting $p = \frac{k-0.5}{N}$.

The algorithm \eqref{eq:combined_averaging_and_local_update_w_noise} is evaluated for different noise variances $\sigma_{\xi}^2$ with the metric \eqref{eq:theorem2_metric}. In Fig. \ref{fig:theorem1_n_2_result_unif_sigma}, the quantile $\theta_{0.89}$ of uniform data was tested. One can see that the estimated states converge toward the true quantile. Fig. \ref{fig:theorem1_n_2_result_unif_sigma_wo_noise} shows the squared error convergence. We use the following metric, since there is no randomness in the absence of communication noise.

\vspace{-4mm}
\small
\begin{align}
\frac{1}{N}\|\mathbf{\bm{\omega}}(i) - \theta_p\mathbf{1}\|^2, \,\,i \ge 0.
\end{align}
\normalsize
One can see that the sequence converges to the true quantile where we experimented with $\theta_{0.01}$, $\theta_{0.49}$, and $\theta_{0.89}$. Note that the $\theta_{0.01}$ and $\theta_{0.49}$ are the minimum and median of $\mathbf{x}$ respectively. The initial trajectories in Fig. \ref{fig:theorem1_n_2_result_unif_sigma_wo_noise} depends on the sensor network structure and the measurement data contained at each node. 

\begin{figure}
	\centering
	\includegraphics[width=1\linewidth]{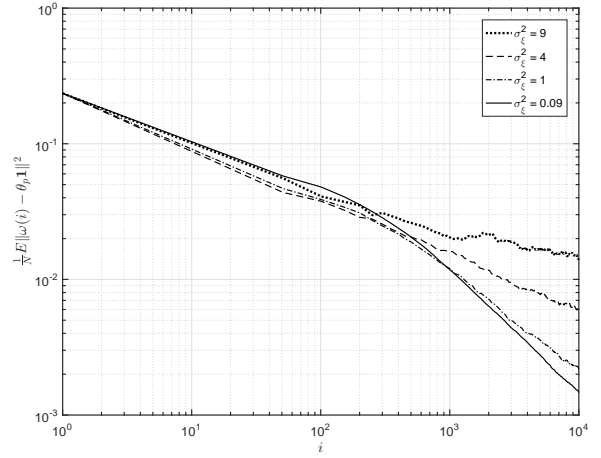}
	\caption{Mean-squared convergence behavior under different communication noise variances.}
	\vspace{-4mm}
	\label{fig:theorem1_n_2_result_unif_sigma}
\end{figure}
\begin{figure}
	\centering
	\includegraphics[width=1\linewidth]{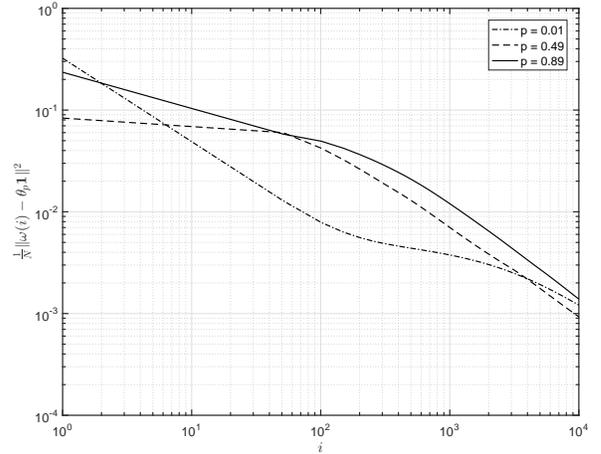}
	\caption{Squared error convergence behavior of various quantiles in the absence of communication noise.}
	\vspace{-4mm}
	\label{fig:theorem1_n_2_result_unif_sigma_wo_noise}
\end{figure}

%

%

\vspace{-1mm}
\subsection{Applications}
\emph{Outlier Identification and Trimmed Mean}: As an application, our algorithm can be used to determine whether individual node measures an outlier value or not. One can judge that larger (or smaller) value than a quantile (e.g., 0.9 or 0.1) is assumed to be outliers. This can be used for robust average consensus, as the estimated mean is not biased by erroneous outliers. Removing outliers is often useful when there exist malicious sensor measurements. If it is identified that the measurement at node $n$ is an outlier, then the node by itself is not averaged with neighboring nodes so that the outliers can be removed when the global average is estimated. An extended application would be the trimmed mean. We often want to average sensor measurements only within the range of $a\% \sim b\%$ where $0 < a < b < 100$. $\theta_{a/100}$ and $\theta_{b/100}$ can be estimated by our algorithm and then individual node can be identified whether they are within the range or not. Then, the trimmed mean is obtained by the average consensus \cite{Boyd2004,Olfati2007} only with the nodes in $[\theta_{a/100},\theta_{b/100}]$. 

\emph{Median Estimation}: A useful metric to measure centrality of sensor measurement data is median. When there are outliers or when the data distribution is skewed, median can be used for a centrality measure of the data. When the data size $N (\ge 2)$ is even, the median can be defined as a value between $x_{0.5N}$ and $x_{0.5N+1}$. By setting $p = 0.5 - \varepsilon/N$, the algorithm \eqref{eq:combined_averaging_and_local_update_w_noise} estimates $x_{0.5N}$, and similarly $x_{0.5N+1}$. When $N$ is odd, the median $\theta_p = x_{\lceil0.5N\rceil}$ is estimated by setting $p = \frac{\lceil0.5N\rceil-\varepsilon}{N}$ where $0<\varepsilon<1$.

\emph{Maximum and Minimum Estimation}: We showed some experimental results of maximum value estimation for the uniform and log-normal data in Fig. \ref{fig:maxconsensus_unif_lognormal}. Minimum value estimation can be obtained similarly by setting $p = \frac{1-\varepsilon}{N}$ where $0 < \varepsilon <1$.

\vspace{-1mm}
\section{Conclusion} \label{conclusion}
We have shown a consensus-based distributed quantile estimation algorithm using empirical CDF with limited size of measurement data. States of a quantile estimation are recursively updated by the combination of \emph{local update} and \emph{averaging} steps in the presence of communication noise. We analyzed convergence behaviors of the algorithm based on mixed-time scale stochastic approximation where the averaging time scale dominates the local update time scale. The estimated state sequence is asymptotically unbiased and converges toward the true quantile in mean-square sense. We demonstrated the performance of algorithm with numerical experiments. Finally, potential applications by using our algorithm were discussed. Maximum, minimum, median, $k$-th smallest element selection out of $N$ elements, outliers identification, and trimmed mean can be obtained in fully distributed sensor networks.

\appendices
\section{Proof for Lemma \ref{lemma:boundedness}} \label{sec:append_proof_lemma:boundedness}
Recall the distributed quantile estimation algorithm \eqref{eq:combined_averaging_and_local_update_w_noise}. Multiplying both sides of \eqref{eq:combined_averaging_and_local_update_w_noise} by $\frac{1}{N} \mathbf{1}^T$ results in 

\vspace{-4mm}
\small
\begin{align}
\omega_{\text{avg}} (i+1) = \omega_{\text{avg}}(i) - \alpha(i) y_{\text{avg}}(i) - \eta(i) \xi_{\text{avg}}(i)
\end{align}
\normalsize
where $\xi_{\text{avg}}(i) \triangleq \frac{1}{N} \mathbf{1}^T \bm{\xi}(i)$ and from \eqref{eq:y_n(i)} 

\vspace{-4mm}
\small
\begin{align}
|y_{\text{avg}}(i)| \triangleq \bigg| \frac{1}{N}\mathbf{1}^T \mathbf{y}(i)\bigg| = \bigg|\frac{1}{N} \sum_{n=1}^{N} u(\omega_n(i) - x_n) - p \bigg| \le 1. \label{eq:y_avg(i)_inequality}
\end{align}
\normalsize
After {}{iterating over} $i$, we have the following stochastic difference equation:

\vspace{-4mm}
\small
\begin{align}
\omega_{\text{avg}} (i+1) = \omega_{\text{avg}}(0) - \sum_{j=0}^{i}\alpha(j) y_{\text{avg}}(j) - \sum_{j=0}^{i}\eta(j) \xi_{\text{avg}}(j). \label{eq:proof_lemma2:w_avg}
\end{align}
\normalsize
{}{By taking $E\left[\cdot\right]$ on both sides of \eqref{eq:proof_lemma2:w_avg}, and due to the inequality of \eqref{eq:y_avg(i)_inequality}, we can obtain}

\vspace{-4mm}
\small
\begin{align}
E\big[\omega_{\text{avg}}(i) - \theta_p\big] \le \omega_{\text{avg}}(0) - \theta_p + \sum_{j=0}^{i-1} \alpha(j),
\end{align}
\normalsize
{}{where the last term of \eqref{eq:proof_lemma2:w_avg} was canceled because $E\big[\xi_{\text{avg}}(i)\big] = 0$ for all $i$.} There exists a decreasing sequence $\eta(i)$ in the form of \eqref{eq:step_size_forms} such that $\limsup_{i \to \infty} \eta(i) \sum_{j=0}^{i-1}\alpha(j) =0$. For example, when $\eta(i) = \frac{1}{(i+1)^{\tau_2}}$ and $\alpha(i) = \frac{1}{(i+1)}$ for $\tau_1=1$ and $0.5< \tau_2 < 1$, there exists $\frac{1}{(i+1)^{\tau_2}} \sum_{j=0}^{i-1}\frac{1}{j+1} < \frac{1}{(i+1)^{\tau_2}} \sum_{j=0}^{i-1} \frac{1}{(j+1)} \frac{(i+1)^{\epsilon}}{(j+1)^{\epsilon}} = \frac{1}{(i+1)^{\tau_2 - \epsilon}} \sum_{j=0}^{i-1}\frac{1}{(j+1)^{1+\epsilon}}$ for all $i>1$ and $0 < \epsilon < 1 - \tau_2 < 0.5$. By Assumption \ref{assumption:stepSizes}, we have

\vspace{-4mm}
\small
\begin{align}
\lim_{i \to \infty}\frac{1}{(i+1)^{\tau_2 - \epsilon}} = 0,\,\,\,\lim_{i\to\infty} \sum_{j=0}^{i-1}\frac{1}{(j+1)^{1+\epsilon}} < \infty.
\end{align}
\normalsize 
Then we can obtain $\limsup_{i \to \infty} \eta(i) \sum_{j=0}^{i-1} \alpha(j) = 0$ for $\tau_1 = 1$ and $0 < \epsilon < 1 - \tau_2 < 0.5$. Therefore, Lemma {}{\ref{lemma:boundedness} is proved.}

\section{Proof for Theorem \ref{theorem:unbiasedness}} \label{sec:append_proof_theorem:unbiasedness}
It is shown that $\|E[\bm{\omega}(i)] - \theta_p \mathbf{1}\|$ converges to 0, as $i \to \infty$. Recall that $\mathbf{L} \cdot \mathbf{1} = \mathbf{0}$. By subtracting $\theta_p \mathbf{1}$ on both sides of \eqref{eq:combined_averaging_and_local_update_w_noise}, it can be rewritten as

\vspace{-4mm}
\small
\begin{align}
\bm{\omega}(i+1) - {\theta_p}\mathbf{1} &= \big(\mathbf{I} - \eta(i)\mathbf{L} \big)  \big(\bm{\omega}(i)- {\theta_p}\mathbf{1} \big)\nonumber \\
&\,\,\,- \alpha(i) \big(\mathbf{I} - \eta(i) \mathbf{L}\big)   \mathbf{y}(i) - \eta(i)\bm{\xi}(i).  \label{eq:theorem1_w(i+1)-theta_1}
\end{align}
\normalsize
Define a rank-1 matrix

\vspace{-4mm}
\small
\begin{align}
\mathbf{G} \triangleq \frac{1}{N}\mathbf{1} \mathbf{1}^T. \label{eq:G_definition}
\end{align}
\normalsize
The average of $\bm{\omega}(i)$ at $i$ can be expressed as 

\vspace{-4mm}
\small
\begin{align}
\mathbf{z}(i) \triangleq \mathbf{G} \bm{\omega}(i) = \omega_{\text{avg}}(i)\mathbf{1}. \label{eq:z(i)definition}
\end{align}
\normalsize
With $\mathbf{z}(i)-\theta_p\mathbf{1} = \mathbf{G}\big(\bm{\omega}(i)-\theta_p\mathbf{1}\big)$, \eqref{eq:theorem1_w(i+1)-theta_1} can be rewritten as

\vspace{-4mm}
\small
\begin{align}
&\bm{\omega}(i+1) - {\theta_p}\mathbf{1} = \Big(\mathbf{I} - \eta(i)\mathbf{R} \Big)  \big(\bm{\omega}(i)- {\theta_p}\mathbf{1} \big) \nonumber \\
&\,\,\,\,\,\,\,\,\,- \alpha(i) \Big(\mathbf{I}  - \eta(i) \mathbf{L}\Big)   \mathbf{y}(i) - \eta(i)\bm{\xi}(i) + \eta(i)\big(\mathbf{z}(i)- {\theta_p}\mathbf{1} \big), \label{eq:theorem1_w(i+1)-theta_1st}
\end{align}
\normalsize
where $\mathbf{R} \triangleq \mathbf{L}+\mathbf{G}$. Taking $E\left[\cdot\right]$ on both sides of \eqref{eq:theorem1_w(i+1)-theta_1st} leads to

\vspace{-4mm}
\small
\begin{align}
&E\big[\bm{\omega}(i+1)\big] - {\theta_p}\mathbf{1} = \Big(\mathbf{I} - \eta(i)\mathbf{R}\Big)  \Big(E\big[\bm{\omega}(i)\big] - {\theta_p}\mathbf{1} \Big) \nonumber \\
&\,\,\,\,\,\,\,\,\,\,- \alpha(i) \Big(\mathbf{I} - \eta(i) \mathbf{L}\Big)   E\big[\mathbf{y}(i)\big] +\eta(i) \Big( E\big[\mathbf{z}(i)\big] - \theta_p\mathbf{1}\Big), \label{eq:E[omega(i)]-thetap_unbiasedness}
\end{align}
\normalsize
where the zero-mean random noise vector $\bm{\xi}(i)$ was canceled {}{due to the expectation}. The smallest eigenvalue of Laplacian matrix $\mathbf{L}$ is equal to zero, and $\|\mathbf{I}-\eta(i)\mathbf{L}\|=1$ for all $i$. Taking $\|\cdot\|$ of both sides of \eqref{eq:E[omega(i)]-thetap_unbiasedness}, by triangle inequality {}{and the property $\|\mathbf{A}\mathbf{b}\| \le \|\mathbf{A}\|\|\mathbf{b}\|$ where $\mathbf{A} \in \mathbb{R}^{N \times N}$ and $\mathbf{b} \in \mathbb{R}^{N \times 1}$}, we have

\vspace{-4mm}
\small
\begin{align}
\Big\|E\big[\bm{\omega}(i+1)&\big] - {\theta_p}\mathbf{1} \Big\| \le \Big\|\mathbf{I} - \eta(i)\mathbf{R} \Big\|   \Big\|E\big[\bm{\omega}(i)\big] - {\theta_p}\mathbf{1} \Big\|\nonumber \\
&\,\,\,\,\, + \alpha(i) \Big\| E\big[\mathbf{y}(i)\big] \Big\| +\eta(i) \Big\| E\big[\mathbf{z}(i)\big] - \theta_p\mathbf{1}\Big\|. \label{eq:norm_E[omega(i)]-thetap_unbiasedness}
\end{align}
\normalsize
{}{$\big\|\mathbf{I} - \eta(i)\mathbf{R}\big\|$ indicates spectral radius which is the maximum eigenvalue of matrix $\mathbf{I} - \eta(i)\mathbf{R}$. For sufficiently large $i$, we have the following inequality:}

\vspace{-4mm}
\small
\begin{align}{}{
\big\|\mathbf{I} - \eta(i)\mathbf{R}\big\| = 1 - \eta(i)\lambda_{\text{min}}(\mathbf{R}) < 1 \label{eq:1-eta(i)lambda}}
\end{align}
\normalsize
{}{where $\lambda_{\text{min}}(\mathbf{R}) = \lambda_{\text{min}}(\mathbf{L}+\mathbf{G}) = \min \left\{\lambda_2(\mathbf{L}), 1\right\}$.} Note that eigenvalues of rank-1 matrix $\mathbf{G}$ are $1,0,\ldots,0$ and the eigenvalues of $\mathbf{L}$ are $0<\lambda_2(\mathbf{L})<\cdots<\lambda_N(\mathbf{L})$. Also, there is a common eigenvector $\mathbf{1}$ for $\mathbf{L}$ and $\mathbf{G}$. In other words, $\big(\mathbf{L}+\mathbf{G}\big)\mathbf{1} = \mathbf{1}$. If $\lambda_2(\mathbf{L}) < 1$, then $\lambda_{\text{min}}(\mathbf{L}+\mathbf{G}) = \lambda_2(\mathbf{L})$. Otherwise, $\lambda_{\text{min}}(\mathbf{L}+\mathbf{G}) = 1$.


%
We simplify the notation and define {}{$\lambda \triangleq \lambda_{\text{min}}(\mathbf{R})$}. By substituting \eqref{eq:1-eta(i)lambda} to \eqref{eq:norm_E[omega(i)]-thetap_unbiasedness} and {}{iterating} over $i-1$, we obtain

\vspace{-4mm}
\small
\begin{align}
&\Big\|E\big[\bm{\omega}(i)\big] - \theta_p \mathbf{1} \Big\| \le \prod_{k=0}^{i-1}\Big(1 - \eta(k)\lambda\Big) \Big\|\bm{\omega}(0) - \theta_p\mathbf{1}\Big\| \nonumber \\
&+\sum_{k=0}^{i-1}\Bigg[\prod_{l=k+1}^{i-1} \Big(1 - \eta(l)\lambda \Big)\Bigg]\, \alpha(k) \Big\|E\big[\mathbf{y}(k)\big]\Big\| \nonumber \\
&+\sum_{k=0}^{i-1}\Bigg[\prod_{l=k+1}^{i-1} \Big(1 - \eta(l)\lambda \Big)\Bigg]\, \eta(k) \Big\| E\big[\mathbf{z}(k)\big] - \theta_p\mathbf{1}\Big\|. \label{eq:unbiasedness:norm_of_expt_omega_thetap}
\end{align}
\normalsize 
We use the property $1-z \le e^{-z}$ for $0 \le z \le 1$. For sufficiently large $k\ge i_0$ there exists positive $\eta(k)\lambda \le 1$. Under Assumption \ref{assumption:stepSizes}, the first term of RHS in \eqref{eq:unbiasedness:norm_of_expt_omega_thetap} goes to zero:

\vspace{-4mm}
\small
\begin{align}
\lim_{i\to\infty}\prod_{k=i_0}^{i-1}\Big(1 - \eta(k) \lambda\Big) \le \lim_{i \to \infty} e^{-\lambda \sum_{k=i_0}^{i-1} \eta(k)} 	= 0.  \label{eq:unbiasedness:1-eta_lambda_goesto_zero}
\end{align}
\normalsize
Since {}{the elements of} $\mathbf{y}(k)$ in \eqref{eq:unbiasedness:norm_of_expt_omega_thetap} are bounded {}{(i.e., $-1 \le {y}_n(k) \le 1$) for all $k$ and $n$} with the step function $u(\cdot)$ and $p$ defined in \eqref{eq:y_n(i)}, we have $\big\|E\big[\mathbf{y}(k)\big]\big\| \le \sqrt{N}$. The second term of RHS in \eqref{eq:unbiasedness:norm_of_expt_omega_thetap} has the following inequality:

\vspace{-4mm}
\small
\begin{align}
\sum_{k=0}^{i-1}\Bigg[\prod_{l=k+1}^{i-1} &\Big(1 - \eta(l)\lambda \Big)\Bigg]\, \alpha(k) \Big\|E\big[\mathbf{y}(k)\big]\Big\| \nonumber \\
&\le \sum_{k=0}^{i-1}\Bigg[\prod_{l=k+1}^{i-1} \Big(1 - \eta(l)\lambda \Big)\Bigg]\, \alpha(k) \sqrt{N}. \label{eq:appendix_proof_th1_2ndTerm}
\end{align}
\normalsize
Thus, we find that the second term of RHS in \eqref{eq:unbiasedness:norm_of_expt_omega_thetap} falls onto the case of \eqref{eq:lemma1:0} in Lemma \ref{lemma1} where the numerators $a_1 = \eta(1)\lambda$ and $a_2 = \alpha(0)\sqrt{N}$ of $r_1(i)$ and $r_2(i)$ in \eqref{eq:lemma1_r1r2}, respectively. {}{Under Assumption \ref{assumption:stepSizes}, $\alpha(k)$ decreases faster than $\eta(l)$ i.e. $\tau_2 < \tau_1$. Then, by Lemma \ref{lemma1}, RHS of \eqref{eq:appendix_proof_th1_2ndTerm}} goes to zero as $i\to\infty$. 

{}{
The third term of RHS in \eqref{eq:unbiasedness:norm_of_expt_omega_thetap} also goes to zero as $i \to \infty$:}

\vspace{-4mm}
\small
\begin{align}
\lim_{i \to \infty} \sum_{k=0}^{i-1}\Bigg[\prod_{l=k+1}^{i-1} \Big(1 - \eta(l)\lambda \Big)\Bigg]  \eta(k) \Big\| E\big[\mathbf{z}(k)\big] - \theta_p\mathbf{1}\Big\| = 0. \label{eq:appendix:3rd_term}
\end{align}
\normalsize
Similarly with \eqref{eq:unbiasedness:1-eta_lambda_goesto_zero}, $\lim_{i \to \infty} \prod_{l=k+1}^{i-1} \Big(1 - \eta(l)\lambda \Big) = 0$ for small $k$. In the case of large $k$, by Lemma \ref{lemma:boundedness},

\vspace{-4mm}
\small
\begin{align}
\lim_{k \to \infty} \eta(k) \Big\| E\big[\mathbf{z}(k)\big] - \theta_p\mathbf{1} \Big\| = 0.
\end{align} 
\normalsize
{}{Therefore, Theorem \ref{theorem:unbiasedness} follows since \small $\lim_{i \to \infty} \Big\|E\big[\bm{\omega}(i)\big] - \theta_p \mathbf{1} \Big\| = 0$ \normalsize in \eqref{eq:unbiasedness:norm_of_expt_omega_thetap}.}
\vspace{-3mm}

\section{Proof for Theorem \ref{theorem:mean-square-convergence}} \label{sec:append_proof_theorem:mean_square_conv}
%
%

{}{We begin with \eqref{eq:theorem1_w(i+1)-theta_1st}. Recall the properties $\big\|\mathbf{I}-\eta(i)\mathbf{L}\big\| = 1$ and $\|\mathbf{A}\mathbf{b}\|^2 \le \|\mathbf{A}\|^2\|\mathbf{b}\|^2$ where $\mathbf{A} \in \mathbb{R}^{N \times N}$ and $\mathbf{b} \in \mathbb{R}^{N \times 1}$. From \eqref{eq:theorem1_w(i+1)-theta_1st} we have}

\vspace{-4mm}
\small
\begin{align}
&\Big\| \bm{\omega}(i+1) - \theta_p \mathbf{1}\Big\|^2 \le \Big\|\mathbf{I} - \eta(i)\mathbf{R}\Big\|^2  \Big\|\bm{\omega}(i)-\theta_p\mathbf{1}\Big\|^2 \nonumber \\
&+\alpha^2(i)\big\|\mathbf{y}(i)\big\|^2 + \eta^2(i)\big\|\bm{\xi}(i)\big\|^2 + \eta^2(i)\Big\|\mathbf{z}(i) - \theta_p\mathbf{1}\Big\|^2 \nonumber \\
&- 2\alpha(i)\bigg[\Big(\mathbf{I} - \eta(i)\mathbf{R} \Big)  \big(\bm{\omega}(i) - {\theta_p}\mathbf{1} \big)\bigg]^T \Big(\mathbf{I} - \eta(i) \mathbf{L}\Big) \mathbf{y}(i) \nonumber \\
&+ 2\eta(i)\bigg[\Big(\mathbf{I} - \eta(i)\mathbf{R} \Big)  \big(\bm{\omega}(i) - {\theta_p}\mathbf{1} \big)\bigg]^T \Big(\mathbf{z}(i)- {\theta_p}\mathbf{1} \Big) \nonumber \\
&-2 \alpha(i) \eta(i) \bigg[ \Big(\mathbf{I} - \eta(i) \mathbf{L}\Big)   \mathbf{y}(i)\bigg]^T \Big(\mathbf{z}(i)- {\theta_p}\mathbf{1} \Big)\nonumber \\
&- 2\eta(i)\bigg[\Big(\mathbf{I} - \eta(i)\mathbf{R} \Big)  \big(\bm{\omega}(i) - {\theta_p}\mathbf{1} \big)\bigg]^T \bm{\xi}(i) \nonumber \\
&+ 2 \alpha(i) \eta(i) \bigg[ \Big(\mathbf{I} - \eta(i) \mathbf{L}\Big)   \mathbf{y}(i)\bigg]^T \bm{\xi}(i) \nonumber \\
&-2 \eta^2(i)\Big(\mathbf{z}(i)- {\theta_p}\mathbf{1} \Big)^T \bm{\xi}(i). \label{eq:norm_sq_mse}
\end{align}
\normalsize
Due to Cauchy-Schwarz inequality and $x \le 1+x^2 $ for any $x \in \mathbb{R}$, the fifth term of \eqref{eq:norm_sq_mse} can be rewritten as 
\vspace{-4mm}
\small
\begin{align}
&- 2\alpha(i)\bigg[\Big(\mathbf{I} - \eta(i)\mathbf{R} \Big)  \big(\bm{\omega}(i) - {\theta_p}\mathbf{1} \big)\bigg]^T \Big(\mathbf{I} - \eta(i) \mathbf{L}\Big) \mathbf{y}(i) \nonumber \\
&\le 2\alpha(i)\Big\|\mathbf{I} - \eta(i)\mathbf{R}  \Big\| \Big\| \bm{\omega}(i) - {\theta_p}\mathbf{1} \Big\| \Big\|\mathbf{y}(i)\Big\| \nonumber \\
&\le 2\alpha(i)\bigg[ 1+ \Big\|\mathbf{I} - \eta(i)\mathbf{R}  \Big\|^2 \Big\| \bm{\omega}(i) - {\theta_p}\mathbf{1} \Big\|^2 \bigg]\Big\|\mathbf{y}(i)\Big\|. \label{eq:norm_sq_mse_cauchy_schwarz}
\end{align}
\normalsize
Similarly the sixth and seventh terms, respectively, can be rewritten as

\vspace{-4mm}
\small
\begin{align}
&2\eta(i)\bigg[\Big(\mathbf{I} - \eta(i)\mathbf{R} \Big)  \big(\bm{\omega}(i) - {\theta_p}\mathbf{1} \big)\bigg]^T \Big(\mathbf{z}(i)- {\theta_p}\mathbf{1} \Big) \nonumber \\
&\le 2\eta(i)\Big\|\mathbf{I} - \eta(i)\mathbf{R}  \Big\| \Big\| \bm{\omega}(i) - {\theta_p}\mathbf{1} \Big\| \Big\|\mathbf{z}(i)- {\theta_p}\mathbf{1} \Big\| \nonumber \\
&\le 2\eta(i)\bigg[ 1+ \Big\|\mathbf{I} - \eta(i)\mathbf{R}  \Big\|^2 \Big\| \bm{\omega}(i) - {\theta_p}\mathbf{1} \Big\|^2 \bigg] \Big\|\mathbf{z}(i)- {\theta_p}\mathbf{1} \Big\| \label{eq:norm_sq_mse_cauchy_schwarz2}
\end{align}
\normalsize
and 

\vspace{-4mm}
\small
\begin{align}
- 2\alpha(i)\eta(i)\bigg[ \Big(\mathbf{I} - &\eta(i) \mathbf{L}\Big)   \mathbf{y}(i)\bigg]^T \Big(\mathbf{z}(i)- {\theta_p}\mathbf{1} \Big)\nonumber \\
&\le 2\alpha(i)\eta(i)\big\| \mathbf{y}(i) \big\| \big\|\mathbf{z}(i)- {\theta_p}\mathbf{1} \big\| . \label{eq:norm_sq_mse_cauchy_schwarz3}
\end{align}
\normalsize
Substituting \eqref{eq:norm_sq_mse_cauchy_schwarz}, \eqref{eq:norm_sq_mse_cauchy_schwarz2}, and \eqref{eq:norm_sq_mse_cauchy_schwarz3} into \eqref{eq:norm_sq_mse} and taking $E[\cdot]$ on both sides of \eqref{eq:norm_sq_mse}, we obtain

\vspace{-4mm}
\small
\begin{align}
&E\bigg[\Big\| \bm{\omega}(i+1) - \theta_p \mathbf{1}\Big\|^2\bigg] \le \Big\|\mathbf{I} - \eta(i)\mathbf{R}\Big\|^2  E\Big[\big\|\bm{\omega}(i)-\theta_p\mathbf{1}\big\|^2\Big] \nonumber \\
&+ \alpha^2(i)E\Big[\big\|\mathbf{y}(i)\big\|^2\Big] + \eta^2(i)E\Big[\big\|\bm{\xi}(i)\big\|^2\Big] \nonumber \\
&+ \eta^2(i)E\Big[\big\|\mathbf{z}(i) - \theta_p\mathbf{1}\big\|^2\Big] + 2\alpha(i) E\Big[\big\|\mathbf{y}(i)\big\|\Big] \nonumber \\
&+ 2\eta(i) E\Big[\big\|\mathbf{z}(i)- {\theta_p}\mathbf{1}\big\| \Big] \nonumber \\
&+ 2\alpha(i) \Big\|\mathbf{I} - \eta(i)\mathbf{R}  \Big\|^2 E\Big[\big\| \bm{\omega}(i) - {\theta_p}\mathbf{1} \big\|^2 \big\|\mathbf{y}(i)\big\| \Big]  \nonumber \\
&+ 2\eta(i) \Big\|\mathbf{I} - \eta(i)\mathbf{R}  \Big\|^2 E\Big[\big\| \bm{\omega}(i) - {\theta_p}\mathbf{1} \big\|^2 \big\|\mathbf{z}(i) - \theta_p \mathbf{1}\big\| \Big]  \nonumber \\
&+ 2\alpha(i)\eta(i) E \Big[\big\| \mathbf{y}(i) \big\| \big\|\mathbf{z}(i)- {\theta_p}\mathbf{1} \big\| \Big] \nonumber \\
&\le \bigg( 1 + 2\alpha(i)E\Big[\big\|\mathbf{y}(i)\big\|\Big] + 2\eta(i) E\Big[\big\|\mathbf{z}(i) - \theta_p \mathbf{1}\big\|\Big]\bigg)\nonumber \\
&\,\,\,\,\,\,\,\,\,\,\,\,\,\,\,\,\,\,\,\,\,\,\,\,\,\,\,\,\,\,\,\,\,\,\,\,\,\,\,\,\,\,\,\,\,\,\,\,\,\,\,\,\,\,\,\,\,\,\,\,\,\,\,\,\,\,\,\,\cdot\Big\|\mathbf{I} - \eta(i)\mathbf{R}\Big\|^2  E\Big[\big\|\bm{\omega}(i)-\theta_p\mathbf{1}\big\|^2\Big] \nonumber \\
&+ \alpha^2(i)E\Big[\big\|\mathbf{y}(i)\big\|^2\Big] + \eta^2(i)E\Big[\big\|\bm{\xi}(i)\big\|^2\Big] \nonumber \\
&+ \eta^2(i)E\Big[\big\|\mathbf{z}(i) - \theta_p\mathbf{1}\big\|^2\Big] + 2\alpha(i) E\Big[\big\|\mathbf{y}(i)\big\|\Big] \nonumber \\
&+ 2\eta(i) E\Big[\big\|\mathbf{z}(i)- {\theta_p}\mathbf{1}\big\| \Big] + 2\alpha(i)\eta(i) E \Big[ \big\| \mathbf{y}(i) \big\| \big\|\mathbf{z}(i)- {\theta_p}\mathbf{1} \big\| \Big] \nonumber \\
&\le \bigg( 1 + 2\alpha(i)\sqrt{N} + 2\eta(i) E\Big[\big\|\mathbf{z}(i) - \theta_p \mathbf{1}\big\|\Big] \bigg) \nonumber \\
&\,\,\,\,\,\,\,\,\,\,\,\,\,\,\,\,\,\,\,\,\,\,\,\,\,\,\,\,\,\,\,\,\,\,\,\,\,\,\,\,\,\,\,\,\,\,\,\,\,\,\,\,\,\,\,\,\,\,\,\,\,\,\,\,\,\,\,\,\cdot\Big\|\mathbf{I} - \eta(i)\mathbf{R}\Big\|^2   E\Big[\big\|\bm{\omega}(i)-\theta_p\mathbf{1}\big\|^2\Big] \nonumber \\
&+\alpha(i) \Big(\alpha(i)N + 2\sqrt{N}\Big) + \eta^2(i)N\sigma_\xi^2 \nonumber \\
& + \eta^2(i)E\Big[\big\|\mathbf{z}(i) - \theta_p\mathbf{1}\big\|^2\Big] + 2\eta(i) E\Big[\big\|\mathbf{z}(i)- {\theta_p}\mathbf{1}\big\| \Big] \nonumber \\
&+ 2\alpha(i)\eta(i)\sqrt{N} E\Big[ \big\|\mathbf{z}(i)- {\theta_p}\mathbf{1} \big\| \Big], \label{eq:norm_sq_cauchy_schwarz_N_sigma}
\end{align}
\normalsize
where {}{the last inequality is due to $\big\|\mathbf{y}(i)\big\| \le \sqrt{N}$ for all $i$ and $\sigma_\xi^2$ denotes noise variance that is bounded as}

\vspace{-4mm}
\small  {}{
\begin{align}
\sigma_\xi^2 = \sup_i E\Big[\big\|\bm{\xi}(i)\big\|^2\Big] \le N d_{\text{max}} \sup_{n,l,i} E \Big[\xi_{nl}^2(i)\Big] < \infty
\end{align}
\normalsize
for $1 \le n,l \le N$, $i \ge 0$, and the maximum degree $d_{\text{max}}$ of network.}

Recall \eqref{eq:1-eta(i)lambda} and $\big(1-\eta(k)\lambda\big)^2 \le 1-\eta(k)\lambda$. Let $\gamma(i) \triangleq 2\alpha(i)\sqrt{N} + 2\eta(i) E\Big[\big\|\mathbf{z}(i) - \theta_p \mathbf{1}\big\|\Big]$ in \eqref{eq:norm_sq_cauchy_schwarz_N_sigma}. After recursions of \eqref{eq:norm_sq_cauchy_schwarz_N_sigma} from 0 up to $i-1$, we have 

\vspace{-4mm}
\footnotesize 
\begin{align}
&E\bigg[\Big\| \bm{\omega}(i) - \theta_p \mathbf{1}\Big\|^2\bigg] \le \prod_{k=0}^{i-1} \Big( 1 + \gamma(k) \Big)\Big(1-\eta(k)\lambda\Big) \big\|\bm{\omega}(0)-\theta_p\mathbf{1}\big\|^2 \nonumber \\
&+\sum_{k=0}^{i-1} \Bigg[\prod_{l=k+1}^{i-1}\Big( 1 + \gamma(l) \Big)\Big(1-\eta(l)\lambda\Big)\Bigg]\alpha(k)\Big(\alpha(k)N + 2\sqrt{N}\Big) \nonumber \\
&+\sum_{k=0}^{i-1} \Bigg[\prod_{l=k+1}^{i-1}\Big( 1 + \gamma(l)\Big)\Big(1-\eta(l)\lambda\Big)\Bigg]\eta^2(k)N\sigma_\xi^2 \nonumber \\
&+\sum_{k=0}^{i-1} \Bigg[\prod_{l=k+1}^{i-1}\Big( 1 + \gamma(l)\Big)\Big(1-\eta(l)\lambda\Big)\Bigg] \eta^2(k)E\Big[\big\|\mathbf{z}(k) - \theta_p\mathbf{1}\big\|^2\Big]  \nonumber \\
&+2 \sum_{k=0}^{i-1} \Bigg[\prod_{l=k+1}^{i-1}\Big( 1 + \gamma(l)\Big)\Big(1-\eta(l)\lambda\Big)\Bigg] \eta(k) E\Big[\big\|\mathbf{z}(k)- {\theta_p}\mathbf{1}\big\| \Big] \nonumber \\
&+2 \sum_{k=0}^{i-1} \Bigg[\prod_{l=k+1}^{i-1}\Big( 1 + \gamma(l)\Big)\Big(1-\eta(l)\lambda\Big)\Bigg] \alpha(k)\eta(k)\sqrt{N} \nonumber \\
&\,\,\,\,\,\,\,\,\,\,\,\,\,\,\,\,\,\,\,\,\,\,\,\,\,\,\,\,\,\,\,\,\,\,\,\,\,\,\,\,\,\,\,\,\,\,\,\,\,\,\,\,\,\,\,\,\,\,\,\,\,\,\,\,\,\,\,\,\,\,\,\,\,\,\,\,\,\,\,\,\,\,\,\,\,\,\,\,\,\,\,\,\,\,\,\,\,\cdot E\Big[ \big\|\mathbf{z}(k)- {\theta_p}\mathbf{1} \big\| \Big] \label{eq:norm_sq_cauchy_schwarz_N_sigma2}
\end{align} 
\normalsize
For sufficiently large $k$ there exists a positive constant $c_1$:

\vspace{-4mm}
\small
\begin{align}
\Big( 1 + \gamma(k)\Big)\Big(1-\eta(k)\lambda\Big) \le 1 - \eta(k)c_1 < 1.\label{eq:theorem2_1-eta(k)c(k)}
\end{align}
\normalsize
Substituting \eqref{eq:theorem2_1-eta(k)c(k)} into \eqref{eq:norm_sq_cauchy_schwarz_N_sigma2}, we can rewrite \eqref{eq:norm_sq_cauchy_schwarz_N_sigma2} as 

\vspace{-4mm}
\small
\begin{align}
&E\Big[\big\| \bm{\omega}(i) - \theta_p \mathbf{1}\big\|^2\Big] \le \prod_{k=0}^{i-1} \Big(1-\eta(k)c_1\Big) \big\|\bm{\omega}(0)-\theta_p\mathbf{1}\big\|^2 \nonumber \\
&+\sum_{k=0}^{i-1} \Bigg[\prod_{l=k+1}^{i-1}\Big(1-\eta(l)c_1\Big)\Bigg]\alpha(k)\Big(\alpha(k)N + 2\sqrt{N}\Big) \nonumber \\
&+\sum_{k=0}^{i-1} \Bigg[\prod_{l=k+1}^{i-1}\Big(1-\eta(l)c_1\Big)\Bigg]\eta^2(k)N\sigma_\xi^2 \nonumber \\
&+\sum_{k=0}^{i-1} \Bigg[\prod_{l=k+1}^{i-1}\Big(1-\eta(l)c_1\Big)\Bigg] \eta^2(k)E\Big[\big\|\mathbf{z}(k) - \theta_p\mathbf{1}\big\|^2\Big] \nonumber \\
&+2\sum_{k=0}^{i-1} \Bigg[\prod_{l=k+1}^{i-1}\Big(1-\eta(l)c_1\Big)\Bigg] \eta(k) E\Big[\big\|\mathbf{z}(k)- {\theta_p}\mathbf{1}\big\| \Big]  \nonumber \\
&+2\sum_{k=0}^{i-1} \Bigg[\prod_{l=k+1}^{i-1}\Big(1-\eta(l)c_1\Big)\Bigg] \alpha(k)\eta(k)\sqrt{N} E\Big[ \big\|\mathbf{z}(k)- {\theta_p}\mathbf{1} \big\| \Big]. \label{eq:norm_sq_cauchy_schwarz_N_sigma3}
\end{align} 
\normalsize
The first term of RHS in \eqref{eq:norm_sq_cauchy_schwarz_N_sigma3} converges to zero as $i \to \infty$ by \eqref{eq:unbiasedness:1-eta_lambda_goesto_zero}. The second, third, fourth, and sixth terms of RHS in \eqref{eq:norm_sq_cauchy_schwarz_N_sigma3} fall onto the case of $\delta_1 < \delta_2$ in Lemma \ref{lemma1}, and they go to zero. {}{The fifth term of RHS in \eqref{eq:norm_sq_cauchy_schwarz_N_sigma3} also goes to zero, similarly with \eqref{eq:appendix:3rd_term}, therefore proving \eqref{eq:theorem2}. }

%



%
%
%

\ifCLASSOPTIONcaptionsoff
  \newpage
\fi



\bibliographystyle{IEEEtran}
\bibliography{references}
%

%








\end{document}